\documentclass[runningheads]{llncs}
\usepackage{graphicx}
\usepackage[pdf]{graphviz}
\usepackage{comment}
\usepackage{bm}
\usepackage{url}
\usepackage{color}
\usepackage{times}
\usepackage{amsmath}
\usepackage{amssymb}
\usepackage{amsfonts}
\usepackage{algorithm}
\usepackage{algpseudocodex}
\usepackage{thm-restate}
\usepackage{booktabs}
\usepackage{multirow}
\usepackage{multicol}
\usepackage{cleveref}

\usepackage{xpatch}
\makeatletter
\newcommand*{\addFileDependency}[1]{% argument=file name and extension
  \typeout{(#1)}
  \@addtofilelist{#1}
  \IfFileExists{#1}{}{\typeout{No file #1.}}
}
\makeatother
\xpretocmd{\digraph}{\addFileDependency{#2.dot}}{}{}

\newcommand{\outline}[1]{}
\newcommand{\defin}[1]{\emph{#1}}
\newcommand*{\functionname}[1]{{{\renewcommand{\rmdefault}{ptm}\fontfamily{ppl}\selectfont\textrm{\textup{#1}}}}}

\crefname{theorem}{Theorem}{Theorem}
\crefname{lemma}{Lemma}{Lemma}
\crefname{corollary}{Corollary}{Corollary}
\crefname{algorithm}{Algorithm}{Algorithm}
\crefname{property}{Property}{Property}
\crefname{figure}{Figure}{Figure}
\crefname{table}{Table}{Table}
\crefname{item}{item}{Item}

\newcommand{\UnaryOperator}[2][]{%
    \ensuremath{\mathop{}\mathopen{}#2\mathopen{}\left(#1\right)}%
}
\DeclareMathAlphabet{\mathup}{OT1}{msb}{m}{n}

\newcommand{\Oh}[1]{\UnaryOperator[#1]{\mathcal{O}}}
\newcommand{\oh}[1]{\UnaryOperator[#1]{o}}

\newcommand{\Om}[1]{\UnaryOperator[#1]{\mathup{\Omega}}}

\newcommand{\calS}{{\mathcal{S}}}
\newcommand{\lcp}[1]{\functionname{LCP}\mathopen{}\left(#1\right)}

\newcommand{\odeg}[1]{d_{\rm o}(#1)}
\newcommand{\ideg}[1]{d_{\rm i}(#1)}

\newcommand{\DFA}{{\mathcal{A}}}

\newcommand{\transfn}{{\delta}}
\newcommand{\startstate}{{r}}
\newcommand{\acceptstate}{{F}}

\newcommand{\heavyedges}{{H}}
\newcommand{\lightedges}{{L}}

\newcommand{\bstary}{{\mathcal{B}}}
\newcommand*{\bst}[1]{{\mathcal{B}\mathopen{}\left({#1}\right)}}

\newcommand{\fid}{{\mathcal{D}}}

\newcommand*{\PADFA}{{\mathsf{A}}}

\newcommand{\dataset}[1]{{\textsf{#1}}}

\begin{document}
\title{
Packed Acyclic Deterministic Finite Automata
}
%
% If the paper title is too long for the running head, you can set
% an abbreviated paper title here
%
\author{Hiroki Shibata\inst{1} \and
Masakazu Ishihata\inst{2} \and
Shunsuke Inenaga\inst{3}
}
\authorrunning{Shibata \and Ishihata \and Inenaga}
% First names are abbreviated in the running head.
% If there are more than two authors, 'et al.' is used.
%
\institute{
Joint Graduate School of Mathematics for Innovation, Kyushu University
\and
NTT Communication Science Laboratories,
\and
Department of Informatics, Kyushu University
}
\maketitle% typeset the header of the contribution
\begin{abstract}
An acyclic deterministic finite automaton (ADFA) is a data structure that represents a set of strings (i.e., a dictionary) and facilitates a pattern searching problem of determining whether a given pattern string is present in the dictionary. 
We introduce the \emph{packed ADFA} (PADFA), a compact variant of ADFA, which is designed to achieve more efficient pattern searching by encoding specific paths as packed strings stored in contiguous memory.
We theoretically demonstrate that pattern searching in PADFA is near time-optimal with a small additional overhead and becomes fully time-optimal for sufficiently long patterns.
Moreover, we prove that a PADFA requires fewer bits than a trie when the dictionary size is relatively smaller than the number of states in the PADFA.
Lastly, we empirically show that PADFAs improve both the space and time efficiency of pattern searching on real-world datasets.

\end{abstract}

%%%%%%%%%%%%%%%%%%%%%%%%%%%%%%%%%%%%%%%%
% Intro
%%%%%%%%%%%%%%%%%%%%%%%%%%%%%%%%%%%%%%%%
\section{Introduction}

%
% Pattern Searching problem and ADFA
% 
Text indexing is a central problem in string processing with many real-world applications, including information retrieval, natural language processing, and bioinformatics.
In this paper, we focus on the \emph{pattern searching problem}, which involves preprocessing a given \emph{dictionary} (a set of distinct strings) into an indexing structure and checking whether a pattern string is contained within the dictionary. 
An \emph{acyclic deterministic finite automaton} (ADFA)~\cite{DBLP:journals/cacm/AppelJ88,DBLP:journals/coling/DaciukMWW00}, also known as a \emph{deterministic acyclic finite state automaton} (DAFSA) or a \emph{directed acyclic word graph} (DAWG), is a fundamental indexing structure for pattern searching.
A \defin{trie}~\cite{DBLP:books/aw/Knuth73} for a dictionary $\calS$ of $k$ strings is an ADFA that forms a rooted tree, and the $k$ paths from its root to an accepting state correspond to strings in $\calS$.
A minimal ADFA (minADFA) for the same $\calS$ can be obtained by merging all isomorphic subtrees of the trie. 
Given a pattern string $P$ of length $m$, both the trie and the minADFA allow us to determine $P \in \calS$ in $\Oh{m \log \sigma}$ time, where $\sigma$ represents the alphabet size, while the minADFA is smaller than the trie.
However, in practice, tries are considered to operate more efficiently than minADFAs.
This is because a trie has a simple tree structure, whereas the minADFA is a directed acyclic graph (DAG), which introduces additional processing overhead.
Therefore, determining which structure offers superior memory efficiency in practical usage remains an open question.

% 
% Packed sting and index structures
% 
Modern computers can process a single word of data in one operation. 
Therefore, when the alphabet size $\sigma$ is sufficiently smaller than the length of a machine word, strings can be processed more efficiently by packing multiple characters into a single machine word, allowing them to be processed in a single operation. 
Typically, a string stores each character in a separate machine word, whereas a packed string stores characters in consecutive memory locations. 
This packing technique enhances memory efficiency, enabling a more compact representation of strings both in terms of storage and processing.
Comparing two packed strings is $\alpha$ times faster than comparing ordinary strings, where $\alpha$ is the number of characters packed into a single word. 
Consequently, the optimal time complexity for pattern searching using packed strings is $\Oh{m/\alpha}$.
Several studies have investigated accelerating pattern searching in tries using packed strings~\cite{DBLP:journals/ieicet/TakagiISA17,DBLP:journals/iandc/TsurutaKKNIBT22};
however, no previous work has applied this technique to DFAs with a theoretical evaluation. 
This research gap arises primarily due to the structural incompatibility between DAG structures and packing techniques.

%
% This paper
%
In this paper, we introduce a \emph{packed ADFA} (PADFA), the first approach to apply the packing technique to ADFA.
Given an ADFA $\DFA$ accepting $\calS$, our proposed method extracts some specific paths, called \emph{heavy paths}, from $\DFA$ by \emph{symmetric centroid path decomposition} (SymCPD).
Then, the method compiles them into a single packed string and the remaining edges as \emph{biased search tree} (BST).
Using the obtained packed string and BST, the method performs the pattern searching in $\Oh{m/\alpha + \log k}$ time.
We theoretically show that a PDFA for any ADFA achieves the time-optimal pattern searching, i.e., $\Oh{m/\alpha}$, if $m$ is sufficiently long compared to $k$.
Additionally, we demonstrate that a PDFA for any minADFA consumes fewer bit of memory than a trie if it has a sufficiently large number of states compared to $k$.
PADFAs are useful not only for pattern searching but also for pattern matching, a task finding a pattern string from a text string.
A DAWG, an ADFA representing all suffixes of the text, is an indexing structure for pattern matching.
Therefore, we can obtain packed DAWG in the same manner as a general ADFA.
The packed DAWG provides the time-optimal pattern matching when $m$ is sufficiently long.
We also conducted experiments with real-world datasets.
The empirical results indicate that PADFA improves both space and time efficiencies of pattern searching.

%%%%%%%%%%%%%%%%%%%%%%%%%%%%%%%%%%%%%%%%
% Contribution
%%%%%%%%%%%%%%%%%%%%%%%%%%%%%%%%%%%%%%%%
\subsection{Contribution}

%
% Pattern Searching
%
We here organize our theoretical contributions.
Let $\calS$ be a set of $k$ distinct strings with the alphabet size $\sigma$, and $\PADFA$ be the packed ADFA, our proposed indexing structure, obtained from an ADFA representing $\calS$ with $n$ states.
We assume the word RAM model with machine words of length $\omega$ and define $\alpha \triangleq \omega / \lceil \log_2 \sigma \rceil$.
Then, the pattern searching problem in $\PADFA$ for a given pattern string $P$ of length $m$ is solved in the following time and space complexity.
\begin{restatable}{theorem}{MainTheoremTime}
\label{thm:PADFA:Time}
The pattern searching in a PADFA for any ADFA takes $\Oh{m/\alpha + \log k}$ time.
\end{restatable}
\begin{restatable}{theorem}{MainTheoremSpace}
\label{thm:PADFA:Space}
The space consumption of a PADFA for any minADFA is $n (1 + \lceil \log_2 \sigma \rceil) + \Oh{k(\log n + \log \sigma)} + \oh{n}$ bits.
\end{restatable}
\noindent 
Thus, $\PADFA$ has an additive overhead of $\log k$ for the time-optimal pattern search $\Oh{m/\alpha}$ in general.
However, by introducing an assumption that $m$ is sufficiently long, the time complexity can be improved as follows.
\begin{restatable}{corollary}{TimeOptimalSearch} 
\label{cor:optimal_time}
Given PADFA $\PADFA$ for any ADFA and any query pattern of length $m \in \Omega(\alpha \log k)$, the pattern searching in $\PADFA$ takes $\Oh{m / \alpha}$ time, which is optimal for pattern searching with packed strings.
\end{restatable}
\noindent
Additionally, we can improve the space complexity by assuming an appropriate condition that $k$ is relatively smaller than $n$ as follows.
\begin{restatable}{corollary}{SpaceWhenKisSmall} 
\label{cor:space_when_k_small}
Given PADFA $\PADFA$ of any minADFA satisfying $\max\{k \log n, k \log \sigma\} \in \oh{n}$, $\PADFA$ consumes $n (1 + \lceil \log_2 \sigma \rceil) + \oh{n}$ bits of space.
\end{restatable}
\noindent
In other words, $\PADFA$ consumes fewer bits of memory than the trie because 
a trie of $n$ vertices requires almost $n (2 + \log_2 \sigma)$ bits, which is at most $n(1 + \lceil \log_2 \sigma \rceil)$ bits.
We believe that the above two conditions are sufficiently realistic. 
This is because situations where more efficient pattern searching is desired typically involve handling large dictionaries and long patterns, and in such cases, we expect that the conditions will generally be met.
Thus, the above theoretical results highlight the advantage of PADFAs over tries.

%
% Pattern matching
%
PADFAs are useful not only for pattern searching but also for substring pattern matching, determining whether a pattern string $P$ of length $m$ occurs in a text string $T$ of length $n$.
It is known that a DAWG of $T$, which is an ADFA representing all $k = n+1$ suffixes of $T$ (including the empty suffix), has only $\Oh{n}$ vertices and edges.
We can construct packed DAWG in the same manner as general ADFAs and derive the following corollary.
\begin{restatable}{corollary}{PackedDAWG} 
\label{cor:packed_dawg}
For any string $T$ of length $n$ and any pattern string $P$ of length $m$, the packed DAWG for $T$ consumes $\Oh{n (\log n + \log \sigma)}$ bits of space and allows substring pattern matching in $\Oh{m / \alpha + \log n}$ time.
\end{restatable}
\noindent
Moreover, using \cref{cor:optimal_time}, we can say that the packed DAWG achieves the optimal time complexity $\Oh{m / \alpha}$ when $m$ is sufficiently long.

%%%%%%%%%%%%%%%%%%%%%%%%%%%%%%%%%%%%%%%%
% Related work
%%%%%%%%%%%%%%%%%%%%%%%%%%%%%%%%%%%%%%%%
\subsection{Related work}

%
% Trie optimization for unary path
% 
In the context of text indexing using graph structures, treating specific paths as strings stored in contiguous memory is a key technique for achieving efficient pattern searching.
\emph{Compact tries} (also known as \emph{patricia tries})~\cite{DBLP:journals/jacm/Morrison68} and \emph{minimal prefix tries}~\cite{DBLP:journals/tse/Aoe89a} handle unary paths in a trie as single edges labeled by strings
These techniques reduce memory consumption and enhance search efficiency, and they can also be applied to ADFA~\cite{fujitaIKMF16_doublearrayCDAWG}.

%
% HPD for trees
%
The \emph{heavy path decomposition} (also known as \emph{centroid path decomposition})~\cite{DBLP:journals/jcss/SleatorT83} is a powerful tool for performing efficient queries on trees by extracting specific paths from the tree.
It handles more edges within packed strings than other methods that optimize only unary paths, while it offers theoretical guarantees on query time.
Some studies have adapted this technique to achieve efficient tries~\cite{DBLP:conf/pods/FerraginaGGSV08,DBLP:journals/jea/KandaKTMF20}.

%
% HPD for DAGs
%
Recently, several variants of heavy path decomposition for general DAGs have been proposed~\cite{DBLP:journals/siamcomp/BilleLRSSW15,DBLP:journals/jacm/GanardiJL21}.
This technique has been applied to various areas, including
efficient random access for grammars~\cite{DBLP:journals/siamcomp/BilleLRSSW15}, and
accelerating query processing for \emph{compacted directed acyclic word graphs} (CDAWGs)~\cite{DBLP:conf/cpm/BilleGS17}.
However, no research has yet applied heavy path decomposition to ADFAs or utilized it to accelerate pattern searching.

%%%%%%%%%%%%%%%%%%%%%%%%%%%%%%%%%%%%%%%%%%%%%%%%%%%%%%%%%%%%%%%%%%%%%%%%%%%%%%%%
% Preliminaries
%%%%%%%%%%%%%%%%%%%%%%%%%%%%%%%%%%%%%%%%%%%%%%%%%%%%%%%%%%%%%%%%%%%%%%%%%%%%%%%%
\section{Preliminaries}

We start by defining the pattern searching problem that this paper addresses. 
After that, we describe the definition of ADFAs, and finally, we present several techniques that serve as key components of the proposed PADFA.

%%%%%%%%%%%%%%%%%%%%%%%%%%%%%%%%%%%%%%%%
% Pattern searching problem on the word RAM model
%%%%%%%%%%%%%%%%%%%%%%%%%%%%%%%%%%%%%%%%
\subsection{The pattern searching problem in the word RAM model}

%
% definitions
%
An alphabet $\Sigma$ is an ordered set of $\sigma$ distinct characters.
A \emph{string} is a finite sequence of characters drawn from $\Sigma$.
For any string $S$, $|S|$ represents its length, and $S[i]$ denotes its $i$th character, where $1 \leq i \leq |S|$.
The empty string, which has a length of zero, is denoted by $\varepsilon$, i.e., $|\varepsilon| = 0$.
Additionally, let $\$$ and $\#$ be special characters:
$\$$ can appear only at the end of non-empty strings, indicating the end of a string, while $\#$ never appears in the input strings and is used solely for the purpose of our algorithm.
For string $S = xyz$, the strings $x$, $y$, and $z$ are referred to as a \emph{prefix}, \emph{substring}, and \emph{suffix} of $S$, respectively.
For $1 \leq i \leq j \leq |S|$, 
$S[i..j] \triangleq S[i] \cdots S[j]$ denotes the substring of $S$ that starts at position $i$ and ends at position $j$.
For convenience, we define $S[i..j] \triangleq \varepsilon$ if $j < i$.
The length of the \emph{longest common prefix} (LCP) of two strings $S$ and $T$
is denoted by $\lcp{S,T} \triangleq \max(\{0\} \cup \{i \mid S[1..i] = T[1..i]\})$.

%
% dictionary and problem (pattern searching)
%
Let $\calS \triangleq \{ S_1, \dots, S_k \}$ be a set of $k$ distinct strings, called a \defin{dictionary}.
The \defin{pattern searching problem} involves evaluating whether 
a given pattern string $P$ is an element of the dictionary $\calS$, i.e., $P \in \calS$.
We consider the indexing version of the above problem.
In other words, $\calS$ is given in advance, and the goal is to preprocess $\calS$ into an appropriate indexing structure for efficiently executing the pattern searching queries for various pattern strings provided online.

%
% word-RAM model
%
In this paper, we focus on the pattern searching problem in the \defin{word RAM model}~\cite{DBLP:journals/jcss/FredmanW93} with a machine word length of $\omega$.
The model provides several constant-time instructions, including random access, word-wise logical and arithmetic operations, and word-wise comparison, which returns the position of the first miss-matched bit.
Let $\ell \triangleq \max_{S\in\calS} |S|$ be the length of the longest string in the dictionary $\calS$.
We assume that $\sigma, k, \ell < 2^\omega$.
Define $\alpha \triangleq \omega / \lceil \log_2 \sigma \rceil$.
Each character is then represented by $\lceil \log_2 \sigma \rceil$ bits, and a string of length $n$ is represented using $n \lceil \log_2 \sigma \rceil$ bits, stored in $\lceil n / \alpha \rceil$ contiguous words.
Strings stored in contiguous memory are called \defin{packed strings}.
Using the constant-time word-wise comparison of the word RAM model, we can obtain $\lcp{X, Y}$ for any packed strings $X$ and $Y$ in $\Oh{|\lcp{X, Y}| / \alpha}$ time.

%%%%%%%%%%%%%%%%%%%%%%%%%%%%%%%%%%%%%%%%
% Acyclic DFAs
%%%%%%%%%%%%%%%%%%%%%%%%%%%%%%%%%%%%%%%%
\subsection{Acyclic deterministic finite automatons (ADFAs)}

%
% algorithm
%
\begin{algorithm}[tb]
    \caption{Determining whether $P \in \calS$ in $\DFA$ accepting $\calS$}
    \label{alg:ADFA-search}
    \begin{algorithmic}[1]
        \State $v_0 \gets r$
        \For{$i$ in $1, \dots, m$}
            \State $v_i \gets \transfn(v_{i-1}, P[i])$ \Comment{move along one edge}
            \State \Return {\bf false} {\bf if} $v_i = \bot$
        \EndFor
        \State \Return $v_m \in \acceptstate$
    \end{algorithmic}
\end{algorithm}

\begin{figure}[tb]
\centering
\includegraphics[width=0.8\textwidth]{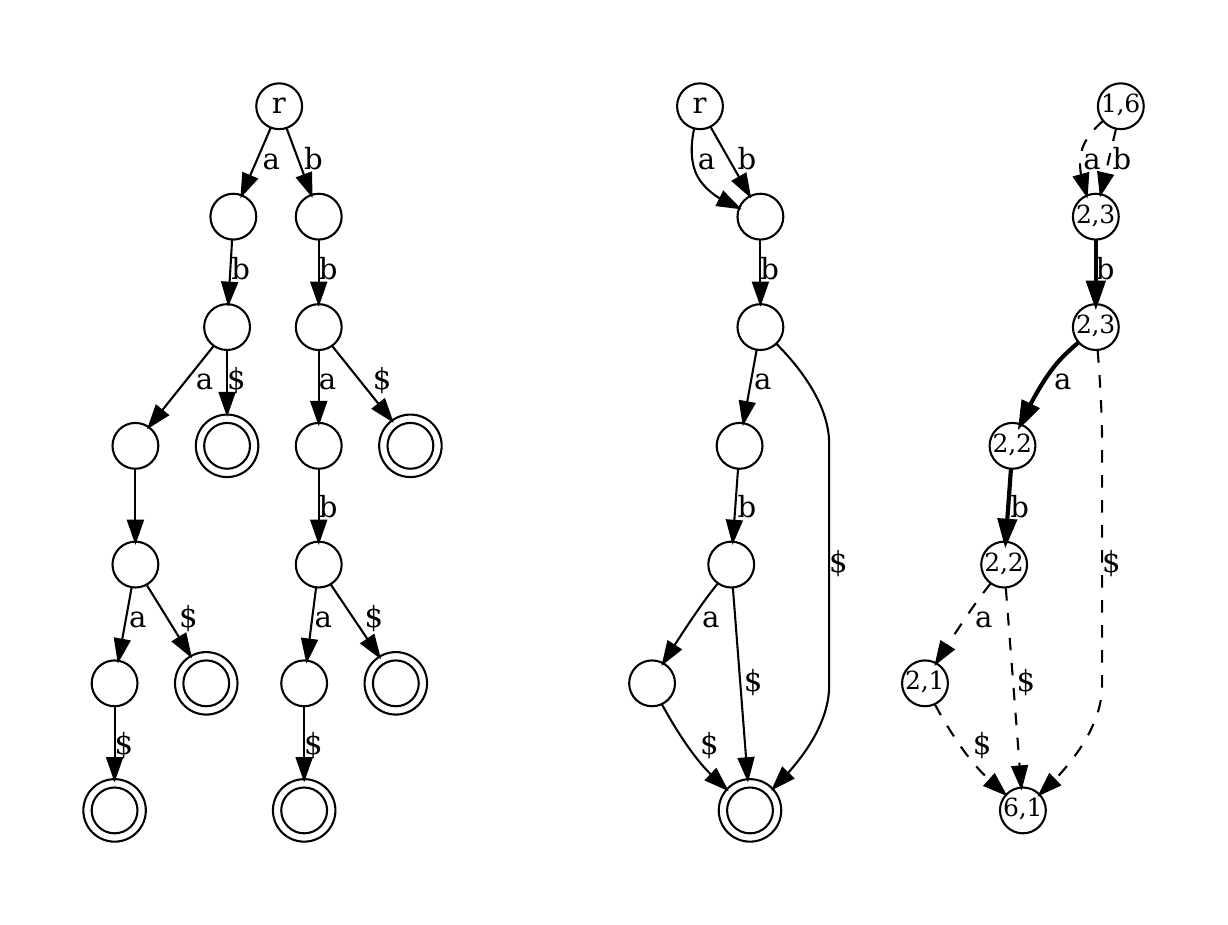} 
\caption{ 
The trie (left), the minADFA (center), and its SymCDP (right), for a dictionary
$\calS = \{ \rm ab\$, abab\$, ababa\$, bb\$, bbab\$, bbaba\$ \}$.
For the trie and minADFA, a vertex labeled $r$ represents a start state, and double circles represent accepting states.
For SymCPD, bold and dashed edges represent heavy and light edges, respectively.
Each vertex is labeled by the two integers indicating $\pi(r, v)$ and $\pi(v, W)$.
}
\label{fig:trie_ADFA_SymCPD}
\end{figure}

%
% FA/ADFA/ADFA
%
A finite automaton (FA) is a tuple $\DFA \triangleq \langle V, E, \acceptstate, \startstate \rangle$, where $V$ is a set of vertices (states), $E \subseteq V \times V \times \Sigma$ is a set of directed labeled edges (transitions), $\acceptstate \subseteq V$ is the set of accepting states, and $\startstate \in V$ is the initial state.
An FA is \emph{deterministic} if the out-edges of the same vertex are labeled with distinct characters.
A deterministic FA (DFA) is \emph{complete} if each vertex has the out-edge labeled by each character from $\Sigma$, and is \emph{partial} otherwise~\cite{DBLP:journals/tcs/BlumerBHECS85}. 
A complete DFA may contain a vertex with no directed path to any accepting state, and an equivalent partial DFA can be obtained by removing such redundant states.
For any partial DFA $\DFA$, let $\transfn : V \times \Sigma \rightarrow V \cup \{ \bot \}$ be its transition function such that $\transfn(u, c) \triangleq v$ if $(u, v, c) \in E$ and $\transfn(u, c) \triangleq \bot$ otherwise; namely, $\DFA$ immediately halts when it reaches $\bot$.
Throughout this paper, we assume that any given DFA is partial and contains no redundant state.
For any $v \in V$, let $\ideg{v}$ and $\odeg{u}$ be the in and out degree of $v$, respectively.
$v$ is \emph{source} if $\ideg{v} = 0$, and $u$ is \emph{sink} if $\odeg{u} = 0$.
A DFA is \emph{acyclic} if it has no cycles.
For any $U \subseteq V$, define $\ideg{U} = \sum_{u \in U} \ideg{u}$ and $\odeg{U} = \sum_{u \in U} \odeg{u}$.
Without loss of generality, we assume that an acyclic DFA (ADFA) has a unique source, which is the initial state $\startstate$.

%
% ADFA
%
An ADFA $\DFA$ \defin{accepts} a string $P$ of legnth $m$ iff there exists a sequence of states $(v_0, \dots, v_m)$ such that $v_0 = \startstate$, $v_m \in \acceptstate$, and $\transfn(v_{i-1}, P[i]) = v_i$ for all $1 \leq i \leq m$. 
Let $\calS$ be the set of all strings accepted by $\DFA$, and then, $\calS$ is finite since $\DFA$ is acyclic and $|V|$ is finite.
\cref{alg:ADFA-search} shows the pattern searching proccess in $\DFA$ and takes $O(m \log \sigma)$ time, independently of $|V|$.
A \defin{trie} for $\calS$ is a tree-formed ADFA accepting $\calS$, where its accepting states correspond to its leaves when every string in $\calS$ ends with $\$$.
Let $\DFA$ be the ADFA obtained by merging isomorphic subtrees of the trie.
Then, $\DFA$ is minimal and has exactly one accepting state, which is its unique sink.
\cref{fig:trie_ADFA_SymCPD} illustrates an example of a trie and its corresponding minimal ADFA (minADFA).

%%%%%%%%%%%%%%%%%%%%%%%%%%%%%%%%%%%%%%%%
% SymCPD
%%%%%%%%%%%%%%%%%%%%%%%%%%%%%%%%%%%%%%%%
\subsection{Building blocks of PADFAs}

%
% abstract
%
We here introduce three techniques,
a \emph{symmetric centroid path decomposition} (SymCPD)~\cite{DBLP:journals/jacm/GanardiJL21},
a \emph{biased search tree} (BST)~\cite{DBLP:journals/siamcomp/BentST85},
and a \emph{fully indexable dictionrry} (FID)~\cite{DBLP:journals/talg/RamanRS07}, 
that are employed to implement our PADFA.

%
% SymCPD
%
SymCPD is a technique for decomposing a DAG into disjoint paths, serving as a generalization of the well-known \defin{heavy path decomposition}~\cite{DBLP:journals/jcss/SleatorT83} used for a tree.
Consider a DAG with a vertex set $V$ and a labeled edge set $E$, assuming a unique source $r \in V$ and a set of sinks $W \subseteq V$.
For any $u, v \in V$, let $\pi(v, u)$ denote the number of directed paths from $v$ to $u$, where $\pi(v, v) \triangleq 1$.
For any $U \subseteq V$, define $\pi(v, U) \triangleq \sum_{u \in U} \pi(v, u)$.
Additionaly, let $\lambda(v) \triangleq \left( \lfloor \log_2 \pi(r, v) \rfloor, \lfloor \log_2 \pi(v, W) \rfloor \right)$.
SymCPD then divides $E$ into two disjoint sets $\heavyedges$ and $\lightedges$, where $\heavyedges \triangleq \{ (u, v, c) \in E \mid \lambda(u) = \lambda(v) \}$ is the set of \defin{heavy edges}, and $\lightedges \triangleq E \setminus \heavyedges$ is the set of \defin{light edges}.
Using these sets, $H$ and $L$, we can derive the following useful properties~\cite[Lemma 2.1]{DBLP:journals/jacm/GanardiJL21}.
\begin{property}
The edge-induced subgraph $\langle V, \heavyedges \rangle$ forms a set of disjoint paths.
\end{property}
\begin{property} \label{prop:SymCPD_lightedge}
Any path on the DAG contains at most $2 \log_2 \lfloor \pi(r, W) \rfloor$ light edges.
\end{property}
The right figure of \cref{fig:trie_ADFA_SymCPD} gives an example of SymCPD.

%
% BST
%
A BST is a data structure storing a subset of an ordered set and various operations, including access, insert, delete, and more.
Consider an ordered universe $\Sigma = \{c_1, \dots, c_\sigma\}$ of $\sigma$ items with the total order $c_1 < \dots < c_{\sigma}$.
Let $w : \Sigma \rightarrow \mathbb{N}_+$ be a weight function over $\Sigma$, and define $w(C) =  \sum_{c \in C} w(c)$ for any $C \subseteq \Sigma$.
We denote by $\bst{C, w}$ a BST storing $C$.
The access operation in $\bst{C, w}$ with a query item $c$, denoted by ${\rm access}(\bst{C, w}, c)$, returns the position index of $c$ if $c \in C$, and NULL otherwise.
$\bst{C, w}$ is implemented as a biased binary search tree, with the following space complexity and time complexity for an access operation.
\begin{property} \label{pro:bst_space}
A succinct representation of $\bst{C, w}$ consumes $(2 + \lfloor \log \sigma \rfloor) |C| + \oh{|C|}$ bits~\cite{DBLP:journals/talg/NavarroS14}. 
\end{property}
\begin{property} \label{pro:bst_time}
An access operation ${\rm access}(\bst{C, w}, c)$ takes $\Oh{1 + \log \left( w(C) / w(c) \right) }$ time if $c \in C$, and $\Oh{\log w(C)}$ time otherwise~\cite{DBLP:journals/siamcomp/BentST85}. 
\end{property}
The factor of 1 that appears in the above time complexity is introduced to avoid the case where $w(C) = w(c)$ results in $\log \left( w(C) / w(c) \right) = 0$, which occurs when $|C| = 1$.
In general, an access operation in a balanced binary search tree requires $\Oh{\log |C|}$ time.
Consequently, by introducing an appropriate weight function $w$, a BST can reduce the average access time for query items provided in an online manner.

%
% FID
% 
A FID is a succinct data structure representing a bit string, a sequence of 0s and 1s, providing the rank and select operations.
Let $\fid$ be a FID representing a bit string $B$ of length $n$. 
Given any $i \leq n$, a rank operation computes the number of 1s in the prefix $B[1,i]$, and a select operation returns the position of the $i$th 1 in $B$.
The space and time complexity of a FID is as follows.
\begin{property} 
\label{pro:fid}
$\fid$ consumes $n + o(n)$ bits and performs constant-time rank and select operations.
\end{property}
%

%%%%%%%%%%%%%%%%%%%%%%%%%%%%%%%%%%%%%%%%%%%%%%%%%%%%%%%%%%%%%%%%
% Packed ADFA
%%%%%%%%%%%%%%%%%%%%%%%%%%%%%%%%%%%%%%%%%%%%%%%%%%%%%%%%%%%%%%%%
\section{Packed ADFAs (PADFA)} \label{se:PADFA}

A PADFA is an implementation of an ADFA that accelerates pattern searching using packed strings.
We begin by defining a PADFA and explaining its pattern-searching process. 
Then, we theoretically analyze the time and space complexities of PADFAs.

%%%%%%%%%%%%%%%%%%%%%%%%%%%%%%%%%%%%%%%%
% Definition
%%%%%%%%%%%%%%%%%%%%%%%%%%%%%%%%%%%%%%%%
\subsection{The definition of PADFAs}

%
% Example of PADFA
%
\begin{figure}[t]
\begin{minipage}{0.4\columnwidth}
\centering
\includegraphics[width=0.55\textwidth]{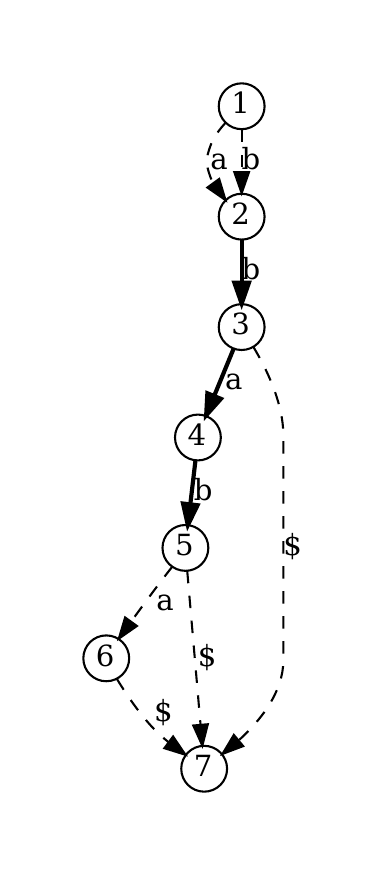} 
\end{minipage}
\begin{minipage}{0.6\columnwidth}
\begin{align*}
\calS &= \{ \rm ab\$, abab\$, ababa\$, bb\$, bbab\$, bbaba\$ \} \\
T &= \mathrm{\#bab\#\#\#}, B = \mathsf{1010110} \\
\Sigma_5 &= \{ a, \$\}, D_5 = (6, 7)
\end{align*}
\caption{The minADFA $\DFA$ accepting $\calS$ (left), and $\calS, T, B, \Sigma_5$ and $D_5$ of the PADFA $\PADFA$ representing $\DFA$ (top-right).
$\PADFA$ has six biased search trees corresponding to the vertices $1, 3, 5$, and $6$.
The vertex $5$ has two light edges $(5, 6, \rm a)$ and $(5, 7, \$)$.
}
\end{minipage}
\label{fig:PADFA}
\end{figure}

%
% notation
%
Consider a dictionary $\calS$ of $k$ distinct strings, where each string in $\calS$ is assumed to end with a special symbol $\$$.
Given an ADFA $\DFA$ accepting $\calS$ with $n$ states, let $\heavyedges$ and $\lightedges$ be the set of heavy and light edges obtained by applying SymCPD to $\DFA$.
Thus, $\DFA$ has at most $k$ sinks, denoted by $W \subseteq V$.
Here, we represent each vertex in $V$ as an integer in $\{ 1, \dots, n\}$
such that $r = 1$ and $v = u + 1$ for any $(u, v, c) \in \heavyedges$.
Such a vertex ordering always exists since $H$ forms disjoint paths.
For any $v \in V$, let $L_v$ denote the light edges out-going from $v$.
Given $L_v$, define $D_v \triangleq \{u \mid (v, u, c) \in L_v\}$ as its distination list, 
and $\Sigma_v \triangleq \{c \mid (v, u, c) \in L_v\}$ as its label list.
Additionally, let $B$ be the binary string of length $n$ such that $B[i] = 1$ iff $L_v \neq \emptyset$, and let $n_L \triangleq \sum_{v \in V} B[v]$.
In other words, $L$ is factorized into $n_L$ disjoint edgesets $\{L_v \mid B[v] = 1\}$.

%
% definition of PDFA
%
A PDFA represents $H$ as a single packed string and $L$ as an array of BSTs with some ancillary data structures. 
Let $T$ be a packed string such that $T[v] = c$ if $(v, v+1, c) \in H$ and $\#$ otherwise; in other words, all heavy edges are packed together in $T$.
Define $w_v : \Sigma \rightarrow \mathbb{N}_+$ such that $w_v(c) = \pi(u, W)$ if $(v, u, c) \in L_v$ and $w_v(c) = 0$ otherwise.
Additionally, let $\bst{v} \triangleq \bst{\Sigma_v, w_v}$ be a BST with universe $\Sigma$.
Then, $L_v$ is represented by a pair of $\bst{v}$ and $D_v$ because the destination of the light edge labeled by $c$ in $L_v$ is stored in $D_v[i]$, where $i = {\rm access}(\bst{v}, c)$.
Define $\bstary \triangleq \{ \bst{v} \mid B[v] = 1\}$ and $D \triangleq \{D_v \mid B[v] = 1\}$.
Also, let $\fid$ be a FID of $B$.
Thus, 
In summary, the PADFA $\PADFA$ of the given $\DFA$ is the quadruplet of the packed string $T$, the BSTs $\bstary$, the destination lists $D$, and the FID $\fid$.
Figure 3 gives an example of PADFAs.
% ここだけ直接参照してるので注意

%
% pattern searching on PADFA
%
\cref{alg:PADFA-search} presents a pattern searching process on $\PADFA$ with a query string $P$. 
In each iteration, $\ell - 1$ indicates the number of characters already matched.
First, the algorithm computes $l = \lcp{T[v..|T|], P[\ell..m]}$ using the packed technique, where $l$ represents the number of characters that can be traversed along the heavy edges. 
After traversing along the heavy edges as long as possible, it checks whether $v$ has a light edge $(v, u, P[\ell])$ using $B$ and $p = {\rm access}(\bst{v}, P[\ell])$.
If such a light edge exists, it retrieves $u$ from $D_v[p]$. 
The process repeats until a termination condition is met.

%
% algorithm
%
\begin{algorithm}[t]
    \caption{Determining whether $P \in \calS$ in $\PADFA$ accepting $\calS$}
    \label{alg:PADFA-search}
    \begin{algorithmic}[1]
        \State $v \gets 1$, $\ell \leftarrow 1$
        \While{$\ell \leq m$}
            \State $l \gets \lcp{T[v..|T|], P[i..m]}$
            \State $v \gets v + l$,  $\ell \gets \ell + l$ \label{line:move_heavyedge}
            \Comment{move along $\ell$ heavy edges}
            \State {\bf break} {\bf if} $\ell > m$
            \State \Return {\bf false} if $B[v] = 0$
            \State $p \gets {\rm access}(\bst{v}, P[\ell])$ \label{line:transfn}
            \State \Return {\bf false} {\bf if} $p = {\rm NULL}$.
            \State $v \gets D_v[p]$, $\ell \gets \ell + 1$ \label{line:move_lightedge}
            \Comment{move along one light edge}
        \EndWhile
        \State \Return {\bf true}
        % \State \Return $(T[v] = \#) \land (B[v] = 0)$
        % \Comment{return true iff $v_i$ is a leaf}
    \end{algorithmic}
\end{algorithm}

%%%%%%%%%%%%%%%%%%%%%%%%%%%%%%%%%%%%%%%%
% Theoretical analysis
%%%%%%%%%%%%%%%%%%%%%%%%%%%%%%%%%%%%%%%%
\subsection{The time and space complexities of PADFAs}

%
% PADFA time complexity
%
Firstly, we discuss the time complexity of \cref{alg:PADFA-search}.
Define $I$ be the number of iterations of the algorithm.
Let $v_i$, $\ell_i$, $l_i$, and $p_i$ be the value of $v$, $\ell$, $l$ and $p$ after \cref{line:move_heavyedge} of the $i$th iteration for any $1 \leq i \leq I$.
Additionaly, define $c_i = P[\ell_i]$ and $u_i = D_{v_i}(p_i)$.
Then, the upper bound of $I$ is given as follows.
%]
\begin{lemma} \label{lem:PADFA:iterations}
The number of iterations $I$ in \cref{alg:PADFA-search} is at most $1 + 2 \lfloor \log_2 k \rfloor$.
\end{lemma}
\begin{proof}
The algorithm alternates between traversing several consecutive heavy edges and a single light edge, so the number of iterations is, at most, the number of light edges in one path plus one.
Since $\pi(r, W) = k$ stands, any path in $\DFA$ contains at most $2 \lfloor \log_2 k \rfloor$ light edges by \cref{prop:SymCPD_lightedge}. 
Consequently, $I \leq 1 + 2 \lfloor \log_2 k \rfloor$.
\qed
\end{proof}
Next, we introduce two lemmas for traversing heavy edges and light edges.
\begin{lemma} \label{lem:PADFA:traversing_heavy}
The total time complexity of traversing heavy edges in \cref{alg:PADFA-search} is at most $\Oh{m / \alpha + \log k}$.
\end{lemma}
\begin{proof}
Since the algorithm traverses at most $m$ heavy edges, it follows that $\sum_{i=1}^I l_i \leq m$ stands.
Given that the time complexity of $\lcp{S_1, S_2}$ is $\Oh{\left\lceil \lcp{S_1, S_2} / \alpha \right\rceil}$, the total complexity obtaining all $l_i$ is $\Oh{I + m / \alpha}$, as derived from the following equation.
\begin{align*}
\sum_{i=1}^I \left\lceil \frac{l_i}{\alpha} \right\rceil 
\leq I + \sum_{i=1}^I \frac{l_x}{\alpha}
= I + \frac{m}{\alpha}
\end{align*}
Since traversing $l_i$ heavy edges can be done in constant time through a simple addition, as shown in \cref{line:move_heavyedge}, the total time complexity for traversing heavy edges is $\Oh{m / \alpha + \log k}$.
\qed
\end{proof}
\begin{lemma} \label{lem:PADFA:traversing_light}
The total time complexity of traversing light edges in \cref{alg:PADFA-search} is at most  $\Oh{\log k}$.
\end{lemma}
\begin{proof}
Because traversing a light edge requires an access operation ${\rm access}(\bst{v_i}, c_i$, the total cost for traversing light edges reaches its maximum when $I$ access operations are executed, with the final access operation returning NULL.
Since $v_{i+1}$ must be a descendants of $u_i$, $\pi(v_{i+1}, W) \leq \pi(u_i, W)$ holds for all $i$.
First, a look-up time of $\bst{v_i}$ from $\bstary$ is constant because the FID $\mathcal{D}$ of $B$ provides a constant-time select operation. 
According to \cref{pro:bst_time}, an access operation ${\rm access}(\bst{v_i}, c_i)$ requires $\Oh{1 + \log \left( w_{v_i}(\Sigma) / w_{v_i}(c_i) \right) } = \Oh{1 + \log \left( \pi({v_i}, W) / \pi(u_i, W) \right) }$ if $p_i$ is not NULL, and $\Oh{\log \pi(v_i, W)}$ otherwise.
Then, the total cost of $I$ access operations is computed as follows.
\begin{align*}
\sum_{i = 1}^{I-1} \left(1 + \log \frac{\pi(v_{i-1}, W)}{\pi(u_{i-1}, W)}\right) + \log \pi(v_m, W)
= I - 1 + \log \frac{\prod_{i = 1}^I \pi(v_i, W)}{\prod_{i=1}^{I-1} \pi(u_i, W)} \\
\leq I + \log \frac{\prod_{i = 1}^I \pi(v_i, W)}{\prod_{i=2}^{I} \pi(v_i, W)} 
= I + \log \pi(v_1, W)
\end{align*}
By combining facts that $\pi(v_1, W) \leq k$ and $I \leq 1 + 2 \lfloor \log_2 k \rfloor$, the total access time is at most $\Oh{\log k}$.
Lastly, the traversing time of a light edge is constant because its destination $u_i$ is obtained by accessing $D_v[p_i]$.
Consequently, the total time complexity of traversing light edges is at most $\Oh{\log k}$.
\qed
\end{proof}
From the two lemmas above, we immediately obtain the following theorem.
\MainTheoremTime*
\noindent
The above theorem shows $\PADFA$ demonstrates pattern searching in near-optimal time, with an overhead of $\Oh{\log k}$.
Assuming $m \in \Om{\alpha \log k}$, we obtain the following corollary.
\TimeOptimalSearch*

% 
% PADFA space complexity
% 
\newcommand{\V}[2]{V_{{\rm #1}, #2}}
Second, we discuss the space complexity of the PADFA for a minADFA.
Let $\DFA$ denote the minADFA accepting $\calS$, and let $\PADFA$ denote the PADFA of $\DFA$.
First, we analyze the number of light edges $L$ in $\PADFA$.
Define $\V{o}{1} \triangleq \{v \in V \mid \odeg{v} = 1\}$ and $\V{o}{2} \triangleq \{v \in V \mid \odeg{v} \geq 2\}$.
In the same manner, we introduce $\V{i}{1}$ and $\V{i}{2}$.
\begin{lemma} \label{lem:PADFA:bound_deg}
For any minADFA, $\odeg{\V{o}{2}}$ and $\ideg{\V{i}{2}}$ are both upper-bounded by $2k$.
\end{lemma}
\begin{proof}
For simplicity, define $n_1 = |\V{o}{1}|$, $n_2 = |\V{o}{2}|$, and $d = \odeg{ \V{o}{2}}$.
Assume a given ADFA $\DFA$ forms a trie, meaning $n = |E| + 1$ and $|W| = k$ hold.
By definition, $n = k + n_1 + n_2$ and $|E| = n_1 + d$ hold.
By combining these facts, we derive $d - n_2 = k - 1$.
Using the fact that $d \geq 2n_2$, $n_2 \leq k - 1$ holds. 
Now, assume $d \geq 2k$.
Then, $n_2 \geq k + 1$ would hold, which contradicts the earlier fact that $n_2 \leq k - 1$.
Therefore, $d < 2k$.
Since a minADFA is obtained by merging some isomorphic subtrees of the trie,
we conclude that $\odeg{\V{o}{\geq 2}} < 2k$ also holds for the minADFA $\DFA$.
Similarly, $\ideg{\V{i}{\geq 2}} < 2k$ can also be proven in the same manner.
\qed
\end{proof}
From the above lemma, we derive the following upper bound on $|L|$.
\begin{lemma} \label{lem:PADFA:bound_light}
For any minADFA, the number of light edges $|L|$ is $\Oh{k}$.
\end{lemma}
\begin{proof}
By the definition of SymCPD, any edge that does not share its starting or ending points with any other edge must be categorized as a heavy edge.
Thus, $|L|$ is bounded by the number of edges that share at least its starting or ending vertex with others.
Consequently, $|L| \leq \odeg{\V{out}{\geq 2}} + \ideg{\V{in}{\geq 2}} \leq 4k$ holds by \cref{lem:PADFA:bound_deg}.
\qed
\end{proof}
Finally, we obtain the following theorem.
\MainTheoremSpace*
\begin{proof}
The PADFA $\PADFA$ consists of four data structures:
the packed string $T$, the BSTs $\bstary$, the destination lists $D$, and the FID $\fid$.
The packed string $T$ consumes $n \lceil \log_2 \sigma \rceil$ bits.
The BSTs $\bstary$ consumes $\Oh{k \log \sigma}$ bits because a single BST $\bst{v}$ consumes $\Oh{|\Sigma_v| \log \sigma }$ bits by \cref{pro:bst_space}, and $\sum_{v} |\Sigma_v| = |L|$, which is upper bounded by $\Oh{k}$.
The destination lists $D$ stores all destinations of $L$ and requires $\Oh{k \log n}$ bits.
Accoding to \cref{pro:fid}, the FID $\fid$ consumes $n + o(n)$ bits.
By summing up all space complexities, the theorem follows.
\qed
\end{proof}
If $k$ is relatively small compared to $n$, the PADFA $\PADFA$ can be represented in a more compact space, as shown below.
\SpaceWhenKisSmall*
\noindent
This corollary demonstrates the advantage of PADFA over trie.
The information-theoretic lower bound of a trie of $n'$ vertices is $n' (2 + \log_2 \sigma)$ bits, which is greater that $n' (1 + \lceil \log_2 \sigma \rceil)$ bits.
Since $n$ is typically smaller than $n'$, $\PADFA$ consumes less memory than the trie, even though $\PADFA$ forms a DAG.
Furthermore, the FID $\fid$ representing $B$ can be compressed using the zeroth-order empirical entropy of each string while still supporting constant-time operations~\cite{DBLP:journals/talg/RamanRS07}.
Since the number of 1s in $B$ is $\Oh{k}$, this representation achieves significant compression when $k$ is sufficiently small with respect to $n$.
By applying this technique, the space complexity of $\PADFA$ can be regarded as $n \lceil \log_2 \sigma \rceil + \oh{n}$ bits.

%
% Packed DAWG
%
Lastly, we show an application of our PADFA for substring pattern matching, determining whether a pattern string $P$ of length $m$ occurs in a text string $T$ of length $n$.
A DAWG of $T$ is an ADFA representing all $n+1$ suffixes of $T$ and consumes $\Oh{n}$ vertices and edges.
A packed DAWG can be obtained in the same manner as general ADFAs and derive the following corollary by directly adapting \cref{thm:PADFA:Time} and \ref{thm:PADFA:Space} to the DAWGs.
\PackedDAWG*
\noindent
In addition, according to \cref{cor:optimal_time}, the packed DAWG achieves the optimal time complexity $\Oh{m / \alpha}$ when $m$ is sufficiently long.

%%%%%%%%%%%%%%%%%%%%%%%%%%%%%%%%%%%%%%%%%%%%%%%%%%%%%%%%%%%%%%%%%%%%%%%%%%%%%%%%
% Experients
%%%%%%%%%%%%%%%%%%%%%%%%%%%%%%%%%%%%%%%%%%%%%%%%%%%%%%%%%%%%%%%%%%%%%%%%%%%%%%%%
\section{Experiments}

We implemented various ADFAs and applied them to multiple real-world datasets to demonstrate that PADFAs achieve better space and time efficiency compared to ADFAs that do not utilize packed strings.
We begin by describing our experimental settings, followed by a presentation of the experimental results.

%%%%%%%%%%%%%%%%%%%%%%%%%%%%%%%%%%%%%%%%
% Experimental settings
%%%%%%%%%%%%%%%%%%%%%%%%%%%%%%%%%%%%%%%%
\subsection{Experimental settings}
\newcommand{\Atrie}{\mathcal{A}_{\rm trie}}
\newcommand{\Apref}{\mathbf{A}_{\rm pref}}
\newcommand{\Apath}{\mathbf{A}_{\rm path}}
\newcommand{\Amini}{\mathcal{A}_{\rm min}}
\newcommand{\Aours}{\mathbf{A}_{\rm min}}

\begin{table}[t]
    \centering
    \caption{The characteristics for each dataset and the size of tries and ADFA.}
    \begin{tabular}{l rrrr rr rr}
        \toprule
        & \multicolumn{4}{c}{dictionary} & \multicolumn{2}{c}{Trie $\Atrie$} & \multicolumn{2}{c}{PADFA $\Aours$} \\
        \cmidrule(lr){2-5}
        \cmidrule(lr){6-7}
        \cmidrule(lr){8-9}
        & $\sigma$ & $k$ & total len.  & ave. len. & $|V|$  & $|E|$  & $|V|$  & $|E|$ \\
        \midrule
        \dataset{url} & 93 & 862,665 & 72,540,387 & 84.089 & 10,146,553 & 10,146,552 & 1,612,336 & 2,040,555  \\
        \dataset{city} & 78 & 177,030 & 1,970,082 & 11.183 & 846,550 & 846,549 & 198,195 & 333,800  \\
        \dataset{prot} & 25 & 157,237 & 46,687,247 & 295.046 & 35,028,185 & 35,028,184 & 32,905,500 & 33,030,196 \\
        \bottomrule
    \end{tabular}
    \label{tab:text_characteristics}
\end{table}
%

%
% dataset
%
We used three real-world datasets, \dataset{url}, \dataset{city}, and \dataset{prot} as input dictionaries.
\dataset{url} consits of URLs from a crawl of the \texttt{.eu} domain condacted in 2005~\cite{weblab,BCSU3}.
\dataset{city} is a list of cities with a population of 500 or more, dumped by GeoNames~\cite{geonames}.
\dataset{prot} contains the first 50 MiB of protein sequences downloaded from the Pizza\&Chilli Corpus~\cite{pizzachilli}.
In all dictionaries, each character was represented using one byte (8 bits).
\cref{tab:text_characteristics} summarizes characteristics of these dictionaries.

We implemented five types of ADFAs: $\Atrie$, $\Apref$, $\Apath$, $\Amini$ and $\Aours$.
$\Atrie$ is a simple trie.
$\Apref$ is a \emph{minimal prefix trie}~\cite{DBLP:journals/tse/Aoe89a} that stores only the minimal prefixes needed to identify each string and represents the remaining suffixes as packed strings.
$\Apath$ is a \emph{path-decomposed trie}~\cite{DBLP:conf/pods/FerraginaGGSV08} that stores heavy paths of $\Atrie$ as packed strings, which is essentially equivalent to the PADFA for a simple trie.
$\Amini$ is the minADFA obtained from $\Atrie$.
$\Aours$ is our PADFA for $\Amini$.
Consequently, $\Atrie$ and $\Amini$ do not use packed strings, while $\Apref$, $\Apath$, and $\Aours$ utilize packed strings.

In the experiments, some data structures described in \cref{se:PADFA} were replaced with more practical alternatives to improve the practical performance of ADFAs.
First, we use \emph{two-stage heavy path decompositioin} (two-stage HPD)~ \cite{DBLP:journals/siamcomp/BilleLRSSW15} instead of SymCPD.
The two-stage HPD also decomposes an edge set into heavy and light edges.
Technically, the sets $H$ and $L$ obtained by the two-stage HPD are slightly different from those obtained by the SymCPD, and the set $L$ in the two-stage HPD is smaller compared to that in the SymPCD. 
However, all theoretical results presented in our papers remain valid when employing the two-stage HPD.
Thus, we used the two-stage HPD in the experiments.
Secondly, we implemented all branches of ADFAs as a simple edge list and accessed them by a simple binary search instead of BSTs.
This change was made because a BST empirically consumes more space and time than a simple edge list.

We first constructed the five types of ADFAs for three input dictionaries and measured their memory consumption.
Next, for each ADFA, we performed pattern searching using all strings in the dictionary as queries and recorded the total computation times.
All programs were implemented in C++ and compiled with GCC 12.2.0 using the -O3 option\footnote{The source code is available at \url{https://github.com/shibh308/Packed_ADFA}.}.
All experiments were performed on a machine running Debian 12, equipped with an Intel(R) Xeon(R) CPU 2.20GHz processor, 32GiB of memory, and a register (word) size $\omega = 64$ bits, meaning that a word can store $\alpha=8$ characters.

%%%%%%%%%%%%%%%%%%%%%%%%%%%%%%%%%%%%%%%%
% Experimental results
%%%%%%%%%%%%%%%%%%%%%%%%%%%%%%%%%%%%%%%%
\subsection{Experimental results}

\begin{table}[t]
    \caption{The memory consumption and the computing times of ADFAs.}
    \centering
    \begin{tabular}{l rr rrr rr rrr}
        \toprule
        & \multicolumn{5}{c}{Memory [MiB]} & \multicolumn{5}{c}{Time [ms]} \\
        \cmidrule(lr){2-6}
        \cmidrule(lr){7-11}
        % &  \multicolumn{3}{c}{Trie} & \multicolumn{2}{c}{PADFA} 
        % &  \multicolumn{3}{c}{Trie} & \multicolumn{2}{c}{PADFA} \\
        % \cmidrule(lr){2-4}
        % \cmidrule(lr){5-6}
        % \cmidrule(lr){7-9}
        % \cmidrule(lr){10-11}
        & $\Atrie$ & $\Apref$ & $\Apath$ & $\Amini$ & $\Aours$
        & $\Atrie$ & $\Apref$ & $\Apath$ & $\Amini$ & $\Aours$ \\
%        & $\BaseTrie$ & $\TailTrie$ & $\PathTrie$ & $\BaseDFA$ & $\PathDFA$ \\
        \midrule
        \dataset{url} & 59.269 & 20.751 & 13.625 & 11.676 & 4.773
        & 5911.507 & 5056.601 & 1046.047 & 5691.689 & 1219.044 \\
        \dataset{city} & 4.945 & 2.045 & 1.618 & 1.910 & 1.270
        & 221.616 & 179.772 & 133.276 & 233.288 & 161.726 \\
        \dataset{prot} & 204.609 & 38.729 & 34.130 & 189.000 & 32.757
        & 2614.497 & 363.963 & 175.183 & 2780.915 & 313.974 \\
        \bottomrule
    \end{tabular}
    \label{tab:memory}
\end{table}
%

%
% Space
%
\cref{tab:memory} shows the memory consumption and computing times of each combination of ADFAs and dictionaries.
For the $\dataset{prot}$, which is composed of long strings, we can see that the $\Atrie$ and $\Amini$, which do not use packed strings, are dramatically larger in size compared to the $\Apref$, $\Apath$, and $\Aours$, which utilize packed strings.
This demonstrates that using packed strings is particularly beneficial for dictionaries composed of long strings.

%
% Time
%
The table also shows that $\Aours$ achieved the best memory efficiency and the second-best time efficiency, while $\Apath$ achieved the second-best memory efficiency and the best time efficiency across all dictionaries.
$\Aours$ and $\Apath$ can be regarded as our PADFAs for the minADFA $\Amini$ and trie $\Atrie$, respectively.
When comparing the pairs $(\Atrie, \Amini)$ and $(\Apath, \Aours)$, we observe that applying our packing technique to tries and minADFAs consistently improves both space and time efficiency.
Furthermore, these results suggest that one can manage the trade-off between time and space efficiency by selecting either a minADFA or trie as input for the PADFA.

%%%%%%%%%%%%%%%%%%%%%%%%%%%%%%%%%%%%%%%%%%%%%%%%%%%%%%%%%%%%%%%%%%%%%%%%%%%%%%%%
% Conclusion
%%%%%%%%%%%%%%%%%%%%%%%%%%%%%%%%%%%%%%%%%%%%%%%%%%%%%%%%%%%%%%%%%%%%%%%%%%%%%%%%
\section{Conclusion}

We proposed PADFA, a general framework for packing any ADFA, which empirically reduces both the time and space complexity of pattern searching.
Theoretically, we demonstrated that pattern searching in a PADFA can be performed in 
$\Oh{m / \alpha + \log k}$ time, achieving time-optimal searching for sufficiently long patterns. 
We also proved that a PADFA constructed from a minADFA consumes less space than a trie when the dictionary size is relatively smaller than the size of the minADFA.
Furthermore, we empirically show that PADFAs for both the trie and minADFA achieved the best space and time efficiency for real-world datasets.
The results also suggest that the generality of PADFA allows for controlling the trade-off between speed and memory in pattern searching.

% ---- Bibliography ----
%
% BibTeX users should specify bibliography style 'splncs04'.
% References will then be sorted and formatted in the correct style.
\newpage
\bibliographystyle{splncs04}
\bibliography{mybib.bib}
\end{document}